\documentclass[12pt,english]{article}
\usepackage[english]{babel}
\usepackage[dvips]{graphicx}
\makeatletter
\usepackage{cite}
\usepackage{paralist}
\usepackage[cmex10]{amsmath}
\usepackage{amssymb}
\usepackage{amsthm}
\usepackage{psfrag}
\usepackage[dvips]{graphicx}
\usepackage{subfig}

\newtheorem{theorem}{Theorem}[section]
\newtheorem{remark}{Remark}[section]

\begin{document}

\title{On the Throughput Capacity of Wireless Multi-hop Networks with ALOHA, Node Coloring and CSMA}

\author{
Salman Malik\footnote{INRIA Paris-Rocquencourt, France. Email: \texttt{salman.malik@inria.fr}}, 
Philippe Jacquet\footnote{INRIA Paris-Rocquencourt, France. Email: \texttt{philippe.jacquet@inria.fr}},
Cedric Adjih\footnote{INRIA Paris-Rocquencourt, France. Email: \texttt{cedric.adjih@inria.fr}},
}
\date{}

\maketitle

\begin{abstract}

We quantify the throughput capacity of wireless multi-hop networks with several medium access schemes. We analyze pure ALOHA scheme where simultaneous transmitters are dispatched according to a uniform Poisson distribution and exclusion schemes where simultaneous transmitters are dispatched according to an exclusion rule such as node coloring and carrier sense based schemes. We consider both no-fading and standard Rayleigh fading channel models. Our results show that, under no-fading, slotted ALOHA can achieve at least one-third (or half under Rayleigh fading) of the throughput capacity of node coloring scheme whereas carrier sense based scheme can achieve almost the same throughput capacity as node coloring.  

\end{abstract}

\section{Introduction}
\label{sec:intro}

Measuring the capacity of random wireless multi-hop networks with various medium access protocols (or schemes) remains a very challenging problem. Seminal work of~\cite{Gupta:Kumar} and later studies, {\it e.g.}, \cite{scaling,scaling2} quantify the capacity in terms of asymptotic scaling laws or bounds. These results may not provide detailed insight into the actual performance of various medium access schemes such as exact achievable capacity or protocol design issues such as any trade-offs involving protocol overheads and protocol performance, {\it etc}. Therefore, we evaluate the {\em throughput capacity} of various medium access schemes and our results provide better insights for performance and overhead trade-off and for designing medium access schemes for wireless multi-hop networks.

Medium access schemes can be broadly classified into two main classes: continuous time access and slotted access. Here, our main focus is on slotted medium access although many of our results can be applied to continuous time medium access. Within slotted medium access category, we distinguish slotted ALOHA, node coloring and carrier sense multiple access (CSMA) schemes. The main goal of this article is to quantify the end-to-end throughput capacity of above mentioned medium access schemes in random wireless multi-hop networks and also see how they compare with each other.

\section{Motivation and Related Works}
\label{sec:context}

In one of the first analyses on capacity of medium access schemes in wireless networks, \cite{Nelson:Kleinrock} studied slotted ALOHA under a very simple geometric propagation model. Using a similar propagation model and assuming that all nodes are within range of each other, \cite{CSMA} evaluated CSMA and compared it with slotted ALOHA in terms of throughput. \cite{Bartek} used simulations to analyze CSMA under a realistic signal to interference ratio (SIR) based interference model and compared it with ALOHA, both slotted and un-slotted. For simulations, \cite{Bartek} assumed that transmitters send packets to their assigned receivers which are located at a fixed distance. 

\cite{Weber,Weber2,Weber3} studied the {\em transmission capacity}, which is the maximum number of successful transmissions per unit area at a specified outage probability, of code division multiple access (CDMA) and ALOHA. They assumed that simultaneous transmitters form a homogeneous Poisson point process (PPP) and used the same model for the location of receivers as in~\cite{Bartek}. The fact that the receivers are not a part of the network (node distribution) model and are located at a fixed distance from the transmitters is a simplification. An accurate model of wireless networks should consider that the transmitters, transmit to receivers which are randomly located in the network. Another related work is \cite{Haenggi} which analyzed local (single-hop) throughput and capacity with slotted ALOHA, in networks with random and deterministic node placement, and time division multiple access (TDMA), in $1D$ line-networks. 

With exclusion schemes, like node coloring or CSMA, correlation between the location of simultaneous transmitters makes it extremely difficult to develop a tractable analytical model. Some of the proposed approaches are as follows. \cite{Guard,Guard2} modeled interferers as PPP and exclude some of the interferers in the guard zone around a receiver. \cite{CSMA-Model,Weber3} used Mat\'ern point process however \cite{Busson} showed that it may lead to an underestimation of the density of simultaneous transmitters and proposed to use Simple Sequential Inhibition ({\em SSI}) or, in case of CSMA, an extension of {\em SSI} called {\em SSI$_k$} point process. But, very few analytical results are available on {\em SSI} and {\em SSI$_k$} point processes and results are usually obtained by simulations.  

It can be noticed that most of the related works are limited to single-hop transmissions only and may not give a realistic view of the actual performance of medium access schemes in wireless multi-hop networks. In this article, we will develop a hybrid model, based on analytical model and Monte Carlo method, to compute the throughput capacity of various medium access schemes in random wireless networks consisting of $N$ nodes. We expect that our results will follow the $O(\sqrt{N/\log N})$ scaling law and will also give additional insight into the constant factors associated with this scaling law. In case of multi-hop networks, some of the related works are as follows. \cite{SR-ALOHA} gave an analysis on the optimal probability of transmission for ALOHA which optimizes the product of simultaneously successful transmissions per unit of space by the average range of each transmission. \cite{Weber4} evaluated the transport capacity of random wireless networks without taking into account any particular routing scheme and under strong assumptions that interferers form a PPP and relays are equally spaced on a straight line between source and destination. 

\section{General Settings}
\label{sec:model}

We consider a random wireless network where $N$ nodes are uniformly distributed on a planar disk of radius $r$. In slotted medium access, at any given slot, simultaneous transmitters in the network are distributed like a set of points, ${\cal S}=\{z_1,z_2,\ldots,z_n$\} where $z_i$ is the location of transmitter $i$. 

Let $\gamma_{ij}$ denote the channel gain from node $i$ to node $j$ such that the received power at node $j$ is $P_i\gamma_{ij}$, where $P_i$ is the transmit power of node $i$. We consider that all nodes use unit transmit power and \mbox{$\gamma_{ij}=F\vert z_i-z_j\vert^{-\alpha}$}, where \mbox{$\alpha>2$} is the attenuation coefficient and $F$ is a random variable of unit mean representing the variation in channel gain because of fading. For the distribution of $F$, we consider the following two cases.
\begin{compactenum}[-]
\item {\em No fading}: $F$ is constant and equal to $1$.
\item {\em Rayleigh fading}: $F$ is exponential with mean equal to $1$. For any packet from node $i$ to node $j$, an independent \mbox{$F=F_{ij,t}$} models the fading on the channel during slot $t$. We assume that this fading factor remains constant over the duration of a slot only.
\end{compactenum}
We also assume that the background noise power is negligible. Therefore, the transmission from node $i$ to node $j$, during slot $t$, is successful only if the following condition is satisfied
$$
\frac{F_{ij,t}\vert  z_i-z_j\vert^{-\alpha}}{I_{j,t}}\geq K~,
$$
where $K$ is the minimum SIR threshold required for successfully receiving the packet and $I_{j,t}$ is the total interference at node $j$, during slot $t$, given by
$$
I_{j,t}=\sum_{k\neq i}\frac{F_{kj,t}}{\vert  z_k-z_j\vert^{\alpha}}~.
$$

\section{Parameters of Interest}
\label{sec:interest}

Our main parameter of interest is the throughput capacity, $\zeta(N)$ which is defined as the expected number of packets delivered to their destinations per slot. The actual throughput capacity depends on the number of nodes, $N$, SIR threshold, $K$, attenuation coefficient, $\alpha$, as well as the following factors:
\begin{compactenum}[-]
\item Medium access scheme employed by the nodes in the network. In \S \ref{sec:models}, we will identify the parameters of various medium access schemes to tune the throughput capacity. 
\item Expected transmission rate of each node, $\Omega_i(N)$. It also depends on the tuning parameters of medium access schemes.
\end{compactenum}

Computing throughput capacity of random wireless network with various medium access schemes is a difficult problem. Here, we will develop an analytical model which we will use with Monte Carlo method to compute the {\em achievable} throughput capacity with various medium access schemes. Note that, we are {\em not} investigating the impact of any routing or queueing strategies. For example, \cite{Stamatiou,Georgiadis} investigated delay minimizing routing and optimal throughput back-pressure routing respectively. In this article, we measure the {\em expected} maximum rate of successful end-to-end transmissions in a random wireless multi-hop network with various medium access schemes and, therefore, our results can compliment the protocol design, {\it e.g.}, designing protocols based on \cite{Stamatiou,Georgiadis}. 

In order to understand the following computation of the throughput capacity in random wireless networks, we assume that the network is operating under these assumptions:
\begin{compactenum}[-]
\item All nodes have an infinite buffer of packets, each destined to a uniformly selected destination. This buffer also carries the routed packets. Therefore, nodes always transmit on the medium when allowed by the medium access scheme employed in the network. As we are only computing the average maximum rate of successful end-to-end transmissions in a random wireless network and routing does not have any impact on the transmission rate of nodes. Therefore, routing protocol is considered beyond the scope of this article.
\item The number of retransmissions per packet are unlimited and we do not take into account any constraints on the expected delay. Therefore, if the random wireless multi-hop network is connected, end-to-end delivery is guaranteed. Note that a random wireless network is connected if there are no isolated nodes in the network as discussed in~\cite{Gupta:Kumar}.
\end{compactenum}

\subsection{Computation of Throughput Capacity}

If $p_{ij}$ is the probability of successful transmission from node $i$ to node $j$, the expected number of transmissions required to deliver a packet from node $i$ to node $j$ is: \mbox{$t_{ij}=1/p_{ij}$}. Note that the value of $t_{ij}$ is an approximation as in case of {\em no fading}, interference may be correlated from one slot to another slot, {\it e.g.}, in case of typical node coloring scheme. However, in our models of medium access schemes, discussed in Section \ref{sec:models}, the state of medium access layer is independent from one slot to another slot and to further minimize the correlation, we also analyze our model under {\em Rayleigh fading}. 

Consequently, the expected minimum number of transmissions required to deliver a packet from node $i$ to node $j$, either directly (single-hop) or over a multi-hop path, is given by
\begin{equation}
m_{ij}=\min_k\Big\{m_{ik}+t_{kj}\Big\}~,~\forall~(i,j)~,
\label{eq:min_tx}
\end{equation}
such that \mbox{$m_{ii}=0$}. 

\begin{theorem}
The throughput capacity of random wireless networks is bounded by,
\begin{equation}
\zeta(N)\leq \frac{(N-1)N\sum_i \Omega_i(N)}{\sum_{ij}m_{ij}}~.
\label{eq:throughput}
\end{equation}
\end{theorem}

\begin{proof}
During $T$ slots, there are on average $T\sum \Omega_{i}(N)$ packet transmission attempts in the network. By the definition of throughput capacity, the number of successfully delivered packets is $\zeta(N) T$. Since each node sends equal traffic to every other node and there are $N(N-1)$ source-destination pairs, the number of packets delivered from a source node $i$ to a destination node $j$ during $T$ slots is $\frac{\zeta(N) T}{N(N-1)}$. We know that the expected minimum number of transmissions required to deliver a packet from node $i$ to node $j$ is $m_{ij}$. Therefore 
$$
\zeta(N)\frac{\sum_{i,j}m_{ij}}{N(N-1)}T~,
$$
is equal to the expected number of transmission attempts in the network during $T$ slots and should be equal to $T\sum \Omega_{i}(N)$. This completes the proof of \eqref{eq:throughput}.
\end{proof}

The challenge is to compute $p_{ij}$ and $m_{ij}$ for all $(i,j)$ and a satisfactory analytical formulation, with all medium access schemes  discussed in \S \ref{sec:models}, is not feasible. Therefore, we will use Monte Carlo method to compute $p_{ij}$ and $m_{ij}$. If node $i$ is isolated, $m_{ij}$, for all $j$, is equal to infinity and $\zeta(N)$ collapses to zero. Therefore, random network of $N$ nodes must be connected and it is a {\em necessary} condition for the feasibility of throughput capacity~\cite{Gupta:Kumar}. A random network of $N$ nodes is connected if $\Omega_i(N)$ scales as $c_1/\log N$ and, consequently, throughput capacity, $\zeta(N)$, scales as $c_2\sqrt{N/\log N}$, for some \mbox{$c_1>0$} and \mbox{$c_2>0$} depending on the medium access scheme, interference model, {\it etc}. Note that \eqref{eq:throughput} incorporates $c_1$ and $c_2$ and, therefore, their values can also be determined. 

\section{Medium Access Schemes}
\label{sec:models}

In this article, we will investigate throughput capacity of random wireless networks with three medium access schemes: node coloring, CSMA and slotted ALOHA. 

In the following discussion, the set of all nodes in the network is ${\cal N}$.

\subsection{Node Coloring Based Schemes}
\label{sec:tdma}

Node coloring schemes use a managed transmission scheme based on TDMA approach. The aim is to minimize the interference between transmissions that cause packet loss. These protocols assign colors to nodes that correspond to periodic slots, {\it i.e.}, nodes that satisfy a spatial condition, either based on physical distance or distance in terms of number of hops, will be assigned different colors. For example, in order to avoid collisions at receivers, all nodes within $k$ hops are assigned unique colors. Typical value of $k$  is $2$. Examples of node coloring schemes are~\cite{unified,SEEDEX,FPRP,DRAND}. 

In this section, instead of considering any particular scheme, we present a model which ensures that transmitters use an exclusion condition in order to avoid the use of same slot within a certain distance. This exclusion distance is defined in terms of euclidean distance $d$ which may be derived from the distance parameter of a typical TDMA-based protocol. Therefore, a slot cannot be shared within a distance of $d$ or, in other words, nodes transmitting in the same slot shall be located at a distance greater or equal to $d$ from each other.  

Following is a model of node coloring schemes which constructs the set of simultaneous transmitters, ${\cal S}$, in each slot (this is supposed to be done off-line so that transmission patterns periodically recur in each slot).
\begin{compactenum}
\item Initialize ${\cal M}={\cal N}$ and ${\cal S}=\emptyset$.
\item Randomly select a node $i$ from ${\cal M}$ and add it to the set ${\cal S}$, i.e, ${\cal S}={\cal S}\cup\{i\}$. 
\item Remove $i$ from the set ${\cal M}$.
\item Remove all nodes from the set ${\cal M}$ which are at distance less than $d$ from $i$.
\item If set ${\cal M}$ is non-empty, repeat from step $2$.
\end{compactenum}
These steps model a centralized or distributed node coloring scheme which {\em randomly} selects the nodes for coloring while satisfying the constraints of euclidean distance. 

Note that the value of $d$ can be tuned at $d^*(N,K,\alpha)$ which maximizes the throughput capacity with node coloring scheme with $N$ nodes, required SIR threshold of $K$ and attenuation coefficient equal to $\alpha$.

\subsection{CSMA Based Schemes}
\label{sec:csma}

We discussed that node coloring schemes require sophisticated medium access algorithm to coordinate simultaneously active transmitters in each slot. On the other hand, CSMA based schemes are simpler but are more demanding on the physical layer. Before transmitting on the channel, a node verifies if the medium is idle by sensing the signal level. If the detected signal level is below a certain threshold, medium is assumed idle and the node transmits its packet  (with the smallest receive to transmit transition time). Otherwise, it may invoke a random back-off mechanism and wait before attempting a retransmission. CSMA/CD (CSMA with collision detection) and CSMA/CA (CSMA with collision avoidance), which is also used in IEEE 802.11, are the modifications of CSMA for performance improvement. 

We adopt a model of CSMA where nodes contend to access medium at the beginning of each slot. In other words, nodes transmit only after detecting that medium is idle. We assume that nodes defer their transmission by a tiny back-off time, from the beginning of a slot, and abort their transmission if they detect that medium is not idle. We also suppose that detection time and receive to transmit transition times are negligible and, in order to avoid collisions, nodes use randomly selected (but different) back-off times. Therefore, the main effect of back-off times is in the production of a random order of the nodes in competition. We use the following simplified construction of the set ${\cal S}$. 
\begin{compactenum}
\item Initialize ${\cal M}={\cal N}$ and ${\cal S}=\emptyset$.
\item Randomly select a node $i$ from ${\cal M}$ and add it to the set ${\cal S}$, {\it i.e.}, ${\cal S}={\cal S}\cup\{i\}$. 
\item Remove $i$ from the set ${\cal M}$.
\item Remove all nodes from the set ${\cal M}$ which can detect a combined interference signal of power higher than $\theta$ (carrier sense threshold), from all transmitters in the set ${\cal S}$, {\it i.e.}, if 
$
\sum_{i\in\cal{S}}F|z_i-z_j|^{-\alpha}\geq \theta~, 
$ 
remove $j$ from $\cal{M}$. Here, $z_i$ is the position of $i$.
\item If set ${\cal M}$ is non-empty, repeat from step $2$.
\end{compactenum}
Above steps model a CSMA based scheme which requires that transmitters do not detect an interference signal of power equal to or higher than $\theta$, during their back-off periods, before transmitting on the medium. At the end of the construction of set ${\cal S}$, some transmitters may experience interference of signal level higher than $\theta$. However, this behavior is in compliance with a realistic CSMA scheme where nodes, which started their transmissions, or, in other words, are already added to the set ${\cal S}$ do not consider the increase in signal level of interference resulting from later transmitters. 

Note that in case of CSMA, $\theta$ can be tuned at $\theta^*(N,K,\alpha)$ to optimize the throughput capacity under given $N$, $K$ and $\alpha$.

\subsection{Slotted ALOHA Scheme}

In slotted ALOHA scheme, nodes do not use any complicated managed transmission scheduling and transmit their packets independently with a certain medium access probability denoted by $p$, {\it i.e.}, in each slot, each node decides independently whether to transmit or otherwise remain silent. 

The throughput capacity of slotted ALOHA at given values of $N$, $K$ and $\alpha$ can be optimized by tuning $p$ at $p^*(N,K,\alpha)$.

\section{Evaluation and Results}
\label{sec:simulations}

It does not appear feasible to build a satisfactory and tractable analytical formulation to derive the accurate values of $p^*(N,K,\alpha)$, $d^*(N,K,\alpha)$ and $\theta^*(N,K,\alpha)$ which maximize the throughput capacity with slotted ALOHA, node coloring and CSMA based schemes respectively. Therefore, we determine these values using the hybrid model of analytical and Monte Carlo methods. We distribute $N$ nodes uniformly over a planar disk of radius $r$. With all medium access schemes, simultaneous transmitters in each slot, {\it i.e.} the set ${\cal S}$, are elected using the models described in \S \ref{sec:models}. In order to compute the probability of successful transmission from each node to every other node, we consider that nodes employ broadcast transmissions and probability of successful transmission is obtained from the number of received and transmitted packets obtained over $25000$ slots. With this information, throughput capacity, $\zeta(N)$, can be computed using \eqref{eq:min_tx} and \eqref{eq:throughput}.

\begin{figure}[!t]
\centering
\includegraphics[scale=0.98]{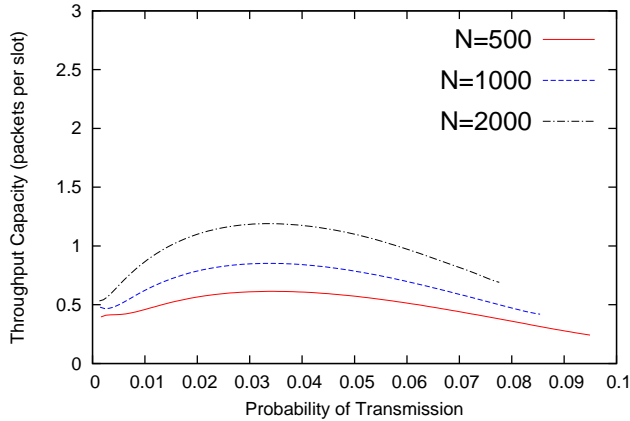}
\caption{Optimization of $\zeta(N)$ with slotted ALOHA scheme under {\em no fading}. $K=20.0$ and $\alpha=4.0$.
\label{fig:optimize_aloha}}
\end{figure}
\begin{figure}[!t]
\centering
\includegraphics[scale=0.98]{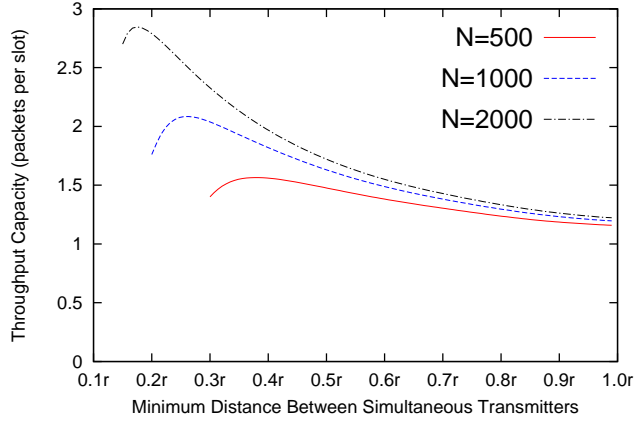}
\caption{Optimization of $\zeta(N)$ with node coloring scheme under {\em no fading}. $K=20.0$ and $\alpha=4.0$.
\label{fig:optimize_tdma}}
\end{figure}
\begin{figure}[!t]
\centering
\includegraphics[scale=0.98]{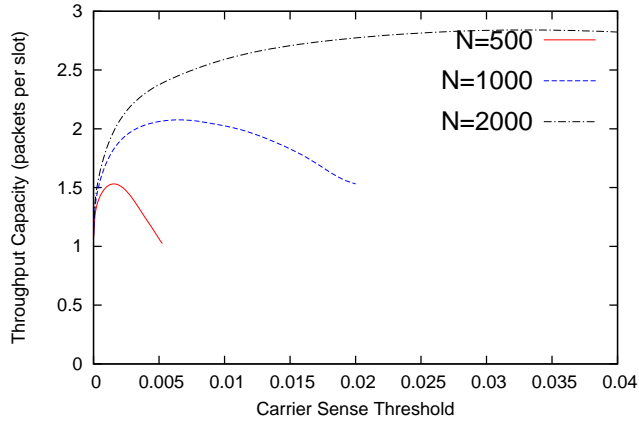}
\caption{Optimization of $\zeta(N)$ with CSMA based scheme under {\em no fading}. $K=20.0$ and $\alpha=4.0$.
\label{fig:optimize_csma}}
\end{figure}

\subsection{Optimization of Throughput Capacity}

In case of slotted ALOHA, we vary the value of probability of transmission, $p$, to determine the optimal $p^*(N,K,\alpha)$ which maximizes $\zeta(N)$ with given values of $N$, SIR threshold, $K$, and attenuation coefficient $\alpha$. Figure \ref{fig:optimize_aloha} shows $\zeta(N)$ with slotted ALOHA averaged over $100$ samples of node distributions. The value of $N$ is varied whereas \mbox{$K=20.0$} and \mbox{$\alpha=4.0$}. It can be observed that as $p$ increases, $\zeta(N)$ increases with rate in $O(\sqrt{N})$ and a maximum occurs at \mbox{$p=p^*(N,20.0,4.0)= c_1/\log N\approx 0.25/\log N$}. Therefore, in a random wireless network, setting the expected transmission rate of each node, $\Omega_i(N)$, equal to $p^*(N,K,\alpha)$ gives the optimal $\zeta(N)$ with slotted ALOHA under given values of $N$, $K$ and $\alpha$. 

Similarly, Fig. \ref{fig:optimize_tdma} and \ref{fig:optimize_csma} show the tuning of node coloring and CSMA based schemes respectively under given values of $N$, \mbox{$K=20.0$} and \mbox{$\alpha=4.0$}. The optimal values of $d^*(N,K,\alpha)$ and $\theta^*(N,K,\alpha)$ can be determined to obtain optimal $\zeta(N)$ with these schemes respectively. Note that, in these cases, $\Omega_i(N)$ is the proportion of slots that each node is expected to be active and transmitting under the given values of $d$ and $\theta$.

\begin{figure*}[!t]
\centering
\subfloat[{\em No fading}.]{
	\includegraphics[scale=0.98]{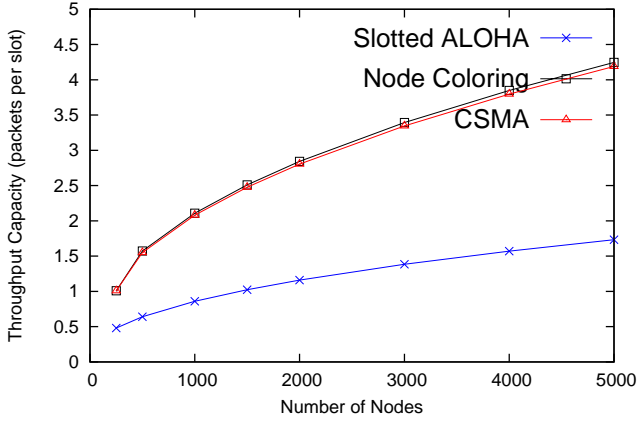}
}
\subfloat[{\em Rayleigh Fading}.]{
	\includegraphics[scale=0.98]{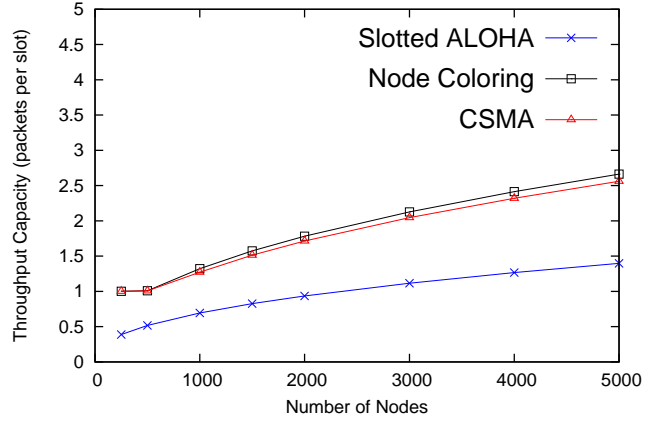}
}
\caption{Optimal $\zeta(N)$ of slotted ALOHA, node coloring and CSMA based schemes. $N$ is varied, $K=20.0$ and $\alpha=4.0$.
\label{fig:comparison_N}}
\end{figure*}
\begin{figure*}[!t]
\centering
\subfloat[{\em No fading}.]{
	\includegraphics[scale=0.98]{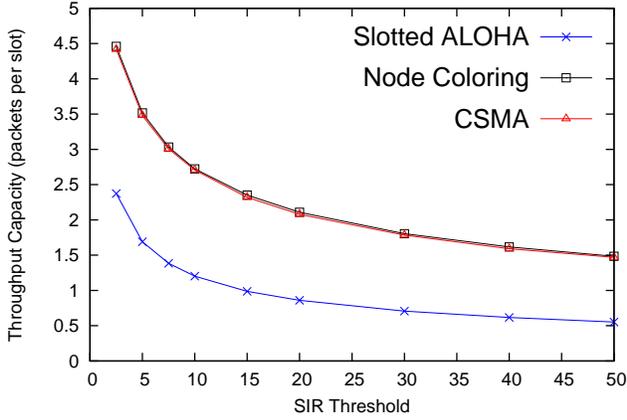}
}
\subfloat[{\em Rayleigh Fading}.]{
	\includegraphics[scale=0.98]{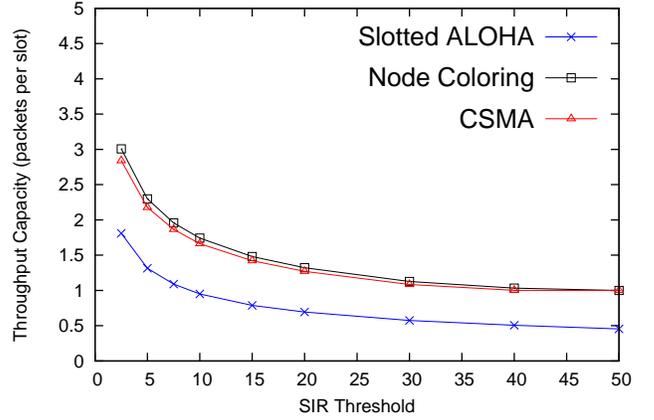}
}
\caption{Optimal $\zeta(1000)$ of slotted ALOHA, node coloring and CSMA based schemes. $N=1000$, $K$ is varied and $\alpha=4.0$.
\label{fig:comparison_K}}
\end{figure*}

\begin{figure*}[!t]
\centering
\subfloat[{\em No fading}.]{
	\includegraphics[scale=0.98]{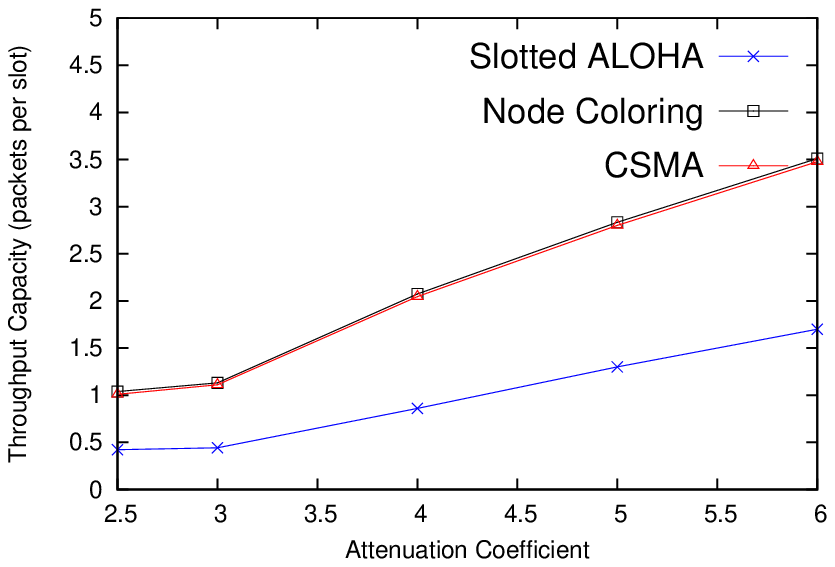}
}
\subfloat[{\em Rayleigh Fading}.]{
	\includegraphics[scale=0.98]{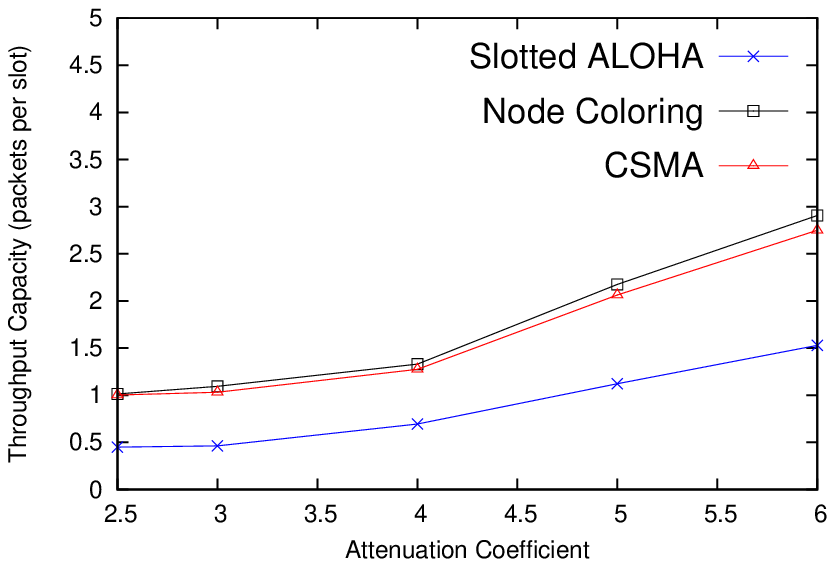}
}
\caption{Optimal $\zeta(1000)$ of slotted ALOHA, node coloring and CSMA based schemes. $N=1000$, $K=20.0$ and $\alpha$ is varied.
\label{fig:comparison_A}}
\end{figure*}

\subsection{Observations}

\begin{remark}
The optimal tuning of all medium access schemes remain insensitive to the fading on the channel. As $N$ increases, the values of $p^*(N,K,\alpha)$, $d^*(N,K,\alpha)$ and $\theta^*(N,K,\alpha)$ are similar under {\em no fading} and {\em Rayleigh fading} of mean $1$. However, under {\em Rayleigh fading} and with low values of $N$, we observed that the parameters of node coloring and CSMA, $d$ and $\theta$ respectively, should be tuned at values such that only $1$ transmitter is active in each slot and therefore achieve higher $\zeta(N)=1$ in this case (see Fig. \ref{fig:comparison_N}(b), $N\leq 500$). As $N$ becomes greater than $500$, $d^*(N,K,\alpha)$ and $\theta^*(N,K,\alpha)$ of node coloring and CSMA are similar under {\em no fading} and {\em Rayleigh fading} channel models.
\end{remark}

\begin{remark}
Fig. \ref{fig:comparison_N} shows $\zeta(N)$ versus $N$ of all schemes at $p^*(N,20.0,4.0)$, $d^*(N,20.0,4.0)$ and $\theta^*(N,20.0,4.0)$ respectively and we can see that as $N$ increases, $\zeta(N)$ scales as $c_2\sqrt{N/\log N}$. In case of slotted ALOHA, under {\em no fading} channel model and \mbox{$K=20.0$} and \mbox{$\alpha=4.0$}, $c_2\approx 0.0715$. Similarly, the constants $c_1$ and $c_2$ can also be determined for node coloring and CSMA based schemes. 
\end{remark}

\begin{remark}
Figures \ref{fig:comparison_K}(a) and \ref{fig:comparison_A}(a) show that under {\em no fading} channel model, slotted ALOHA can achieve {\em at least} one-third or more of the throughput capacity of node coloring and CSMA based schemes. However, under {\em Rayleigh fading} of mean $1$, {\it i.e.} Fig. \ref{fig:comparison_K}(b) and \ref{fig:comparison_A}(b), this factor improves to one-half or even less. From Fig. \ref{fig:comparison_K} and \ref{fig:comparison_A}, we can see that {\em Rayleigh fading} of mean $1$ reduces $\zeta(1000)$ with slotted ALOHA by \mbox{$10\sim25\%$} when compared with {\em no fading} case. In contrast, fading reduces $\zeta(1000)$ of node coloring and CSMA by approximately $35\%$. 
\end{remark}

\begin{remark}
Figures \ref{fig:comparison_N}, \ref{fig:comparison_K} and \ref{fig:comparison_A} show that throughput capacity of CSMA based scheme is slightly lower than node coloring scheme. The reason is that CSMA uses exclusion rule base on carrier sense and this results in lower density (packing) of simultaneous transmitters as compared to node coloring scheme. This can also be observed from the comparison of results on packing densities with $SSI$ and $SSI_k$ point processes in \cite{Busson}. Moreover, under {\em Rayleigh fading} of mean $1$, this difference is exacerbated but by a very small margin.
\end{remark}

\section{Conclusions}
\label{sec:conclude}

We evaluated throughput capacity of random wireless networks with various medium access schemes. Because of the lack of any satisfactory analytical model, we used a hybrid of analytical and Monte Carlo methods to evaluate slotted ALOHA, node coloring and CSMA based schemes. Our results show that compared to slotted ALOHA, node coloring can increase throughput capacity by a factor of $3$ (or less). However, this factor reduces to only $2$ (or less) under more realistic channel model with Rayleigh fading. Moreover, CSMA can achieve almost similar capacity as node coloring. 

The conclusion of this work is that improvements above slotted ALOHA are limited in performance and may be costly in terms of protocol overheads, {\it e.g.}, node coloring schemes which may require significant protocol overhead. Our results show that incorrect tuning of medium access schemes under given system parameters can lead to significantly degraded performance. Therefore, efforts should be diverted to optimizing existing simpler medium access schemes and also to designing efficient routing strategies. Our results are also relevant when nodes move according to an i.i.d. mobility process such that, at any time, the distribution of nodes in the network is homogeneous. However, designing medium access schemes for mobile networks, {\it e.g.}, node coloring schemes may require additional overhead. In the sequel, we will investigate the optimal medium access scheme in wireless networks and we will also compare it with the schemes discussed here.  

\bibliographystyle{hieeetr}
\bibliography{multihop_capacity_arxiv}

\end{document}